\documentclass{glabbeek}
\usepackage{breakurl}           
\usepackage{latexsym}   	

\newfont{\bbb}{bbm10 scaled 1100}  
\newfont{\bbbs}{bbm10 scaled 800}  

\newtheorem{defi}{Definition}
\newtheorem{theo}{Theorem}
\newtheorem{prop}{Proposition}
\newtheorem{lemm}{Lemma}
\newtheorem{coro}{Corollary}
\newtheorem{conj}{Conjecture}
\newtheorem{exam}{Example}

\newenvironment{definition}[1]{\begin{defi} \rm \label{df-#1} }{\end{defi}}
\newenvironment{theorem}[1]{\begin{theo} \rm \label{thm-#1} }{\end{theo}}

\newenvironment{corollary}[1]{\begin{coro} \rm \label{cor-#1} }{\end{coro}}
\newenvironment{conjecture}[1]{\begin{conj} \rm \label{con-#1} }{\end{conj}}

\newenvironment{proof}{\begin{trivlist} \item[\hspace{\labelsep}\bf Proof:]}{\hfill $\Box$\end{trivlist}}

\newenvironment{itemise}{\begin{list}{$\bullet$}{\leftmargin 18pt
                        \labelwidth\leftmargini\advance\labelwidth-\labelsep
                        \topsep 0pt \itemsep 0pt \parsep 0pt}}{\end{list}}

\newcommand{\df}[1]{Definition~\ref{df-#1}}
\newcommand{\thm}[1]{Theorem~\ref{thm-#1}}

\newcommand{\cor}[1]{Corollary~\ref{cor-#1}}

\newcommand{\sect}[1]{Section~\ref{sec-#1}}

\newcommand{\plat}[1]{\raisebox{0pt}[0pt][0pt]{#1}}     
\newcommand{\goto}[1]{\stackrel{#1}
	{\raisebox{0pt}[3pt][0pt]{$\longrightarrow$}}}	
\newcommand{\gonotto}[1]{\hspace{4pt}\not\hspace{-4pt}	
	\stackrel{#1\ }{\longrightarrow}}
\newcommand{\dto}[1]{\mathrel{\stackrel{#1\ }         	
	{\raisebox{0pt}[4pt][0pt]{$\Longrightarrow$}}}}	
\newcommand{\infr}[2]                                   
{\rule{0mm}{6mm} \begin{array}{c} #1\\[0.1ex] \hline \rule{0ex}{2.7ex}#2 \end{array}}
\newcommand{\mathbb}[1]{\mbox{\bbb #1}}                 
\newcommand{\pow}{{\cal P}}			        
\newcommand{\Act}{{\it A}}                              
\newcommand{\ptr}{{\it ptraces}}                        
\newcommand{\traces}{{\it traces}}
\newcommand{\deadlocks}{{\it deadlocks}}
\newcommand{\ct}{{\it CT}}                              
\newcommand{\infinite}{{\it inf}}                       
\newcommand{\infd}{{\it inf\!\!}_\bot}
\newcommand{\infdd}{{\it inf\!\!}_d}
\newcommand{\ini}{{\it initials}}
\newcommand{\failures}{{\it failures}}
\newcommand{\faild}{{\it failures\!}_\bot}
\newcommand{\faildd}{{\it failures\!}_d}
\newcommand{\diverg}{{\it divergences}}
\newcommand{\divd}{{\it divergences\!}_\bot}
\newcommand{\act}{\mbox{\sc action}}
\newcommand{\eff}{\mbox{\sc effect}}
\renewcommand{\phi}{\varphi}

\title{The Coarsest Precongruences\\ Respecting Safety and Liveness Properties}
\author{Rob van Glabbeek
\institute{NICTA, Sydney, Australia}\vspace{-2pt}
\institute{School of Computer Science and Engineering,
University of New South Wales, Sydney, Australia}}

\def\titlerunning{Full Abstraction for Safety and Liveness Properties}
\def\titlerunning{The Coarsest Precongruences Respecting Safety and
  Liveness Properties}
\def\authorrunning{R.J. van Glabbeek}
\begin{document}
\maketitle

\begin{abstract}
This paper characterises the coarsest refinement preorders on labelled
transition systems that are precongruences for renaming and partially
synchronous interleaving operators, and respect all safety, liveness,
and conditional liveness properties, respectively.
\end{abstract}

\section{Introduction}

The goal of this paper is to define and characterise certain semantic
equivalences $\equiv$ and refinement preorders $\sqsubseteq$ on
processes.  The idea is that $p\equiv q$ says, essentially, that for
practical purposes processes $p$ and $q$ are equally suitable, i.e.\
one can be replaced for by the other without untoward side
effects. Likewise, $p \sqsubseteq q$ says that for all practical
purposes under consideration, $q$ is at least as suitable as $p$,
i.e.\ it will never harm to replace $p$ by $q$. Thus, one should have
that $p \equiv q$ iff both $p \sqsubseteq q$ and $q \sqsubseteq p$.

Naturally, the choice of $\equiv$ and $\sqsubseteq$ depends on how one
models a process, and what range of practical purposes one
considers. I this paper I restrict myself to one of the most basic
process models: \emph{labelled transition systems}. I study processes
that merely perform actions $a$, $b$, $c$, \ldots which themselves are
not subject to further investigations. These actions may be
instantaneous or durational, but they may not last forever; moreover,
in a finite amount of time only finitely many actions can be carried
out.  I distinguish between \emph{visible} actions, that can be
observed by the environment of a process, and whose occurrence can be
influenced by this environment, and \emph{invisible} actions, that
cannot be observed of influenced. Since there is no need to
distinguish different invisible actions, I can just as well consider
all of them to be occurrences of the same invisible action, which is
traditionally called $\tau$. Furthermore, I abstract from real-time
and probabilistic aspects of processes.

This choice of process model already rules out many practical purposes
for which one process could be more suitable than another. I can for
instance not compare processes on speed, since this is an issue that
my process model has already abstracted from. In fact, the only
aspects of processes that are captured by such a model and that may
matter in practical applications, are the sequences of actions that a
process may perform in a, possibly infinite, run, performed in, or in
collaboration with, a certain environment. As the invisible action is
by definition unobservable, it moreover suffices to consider sequences
of \emph{visible} actions. A sequence of visible actions that a
process $p$ may perform is called a \emph{trace} of $p$; it is a
\emph{complete trace} if it is performed during a maximal run of $p$,
one that cannot be further extended. Obviously, the traces of $p$ are
completely determined by the complete traces of $p$, namely as their
prefixes.

Based on the considerations above, it is tempting to postulate that
the relevant behaviour of a process, as far as discernible in terms of
labelled transition systems, is completely determined by its set of
complete traces; hence two processes should be equivalent if they have
the same competed traces. However, this argument bypasses the role of
the environment in influencing the behaviour of a process. Often, one
allows the actions a process performs to be synchronisations with the
environment, and the environment can influence the course of action of
a process by synchronising with some actions but not with others.
Therefore, a safe over-approximation of the relevant behaviour of
a process is not merely its set of complete traces, but rather its set
of complete traces obtained as a function of the environment the
process is running in.

In this paper I consider a neutral environment in which all courses of
action are possible and the behaviour of a process is indeed
determined by its complete traces. All other ways in which the
environment may influence the behaviour of a process are given in terms
of \emph{contexts} build from other processes and certain composition
operators. It could for instance be that in the neutral environment
there is no way to tell the difference between processes $p$ and\ $q$;
maybe because they have the same set of complete traces. However, for
a suitable parallel composition operator $\|$ and other process $r$ it
may be that there is a manifest practical difference between $p\|r$
and $q\|r$, so that one has $p\|r \not\equiv q\|r$. Now that fact
alone is taken to be enough reason to postulate that $p \not\equiv q$.
Namely the difference between $p$ and $q$ can be spotted by placing
them in a context $\_\!\_\|r$. This context can be regarded as an
environment in which the behaviours of $p$ and $q$ differ.

Following this programme, a suitable semantic equivalence on processes
is defined in terms of two requirements. First of all the behaviour of
processes is compared in the neutral environment. This entails
isolating a class $\cal C$ of properties $\phi$ of processes that are
deemed relevant in a given range of applications. One then requires
for $p \equiv q$ to hold that $p$ and $q$ have the same properties
from this class:
\begin{equation}\label{respect}
p \equiv q ~~\Rightarrow~~ \forall \phi\in {\cal C}.\; (p \models \phi
\Leftrightarrow q \models \phi)
\end{equation}
where $p\models\phi$ denotes that process $p$ has the property $\phi$.
An equivalence $\equiv$ that satisfies this last requirement is said
to \emph{respect} or \emph{preserve} the properties in $\cal C$.
The second requirement entails selecting a class $\cal O$ of useful
operators for combining processes. One then requires that for any
context $C[\_\!\_]$ (such as $\_\!\_\| r$) built from operators from
$\cal O$ and arbitrary processes, that
\begin{equation}\label{congruence}
p \equiv q ~~\Rightarrow~~ C[p] \equiv C[q].
\end{equation}
An equivalence $\equiv$ that satisfies this last requirement is called
a \emph{congruence} for $\cal O$.
For the sake of intuition it may help to consider the contrapositive
formulation of these implications: if there exists a property $\phi$
in $\cal C$ that holds for $p$ but not for $q$, or vice versa, then
$p$ and $q$ cannot be considered equivalent. Likewise $C[p] \not\equiv
C[q]$ implies $p \not\equiv q$.

These two requirements merely insist that the desired equivalence
$\equiv$ does not identify processes that in some context differ on
their relevant properties. They are satisfied by many equivalence
relations, including the identity relation, that distinguishes all
processes. In order to characterise precisely when two systems have
the same relevant properties in any relevant context, one takes the
\emph{coarsest} equivalence satisfying (\ref{respect}) and
(\ref{congruence}); the one making the most identifications.  This
equivalence is called \emph{fully abstract} w.r.t.\ $\cal C$ and $\cal
O$. It always exists, and, as is straightforward to check, is characterised by
$$p \equiv q ~~\Leftrightarrow~~ \forall {\cal O}\mbox{-context } C[\_\!\_].\;
\forall\phi\in {\cal C}.\;  (C[p] \models \phi \Leftrightarrow C[q] \models \phi).$$

When, for a certain application, the choice of $\cal C$ and $\cal O$
is clear, the unique equivalence relation that is fully abstract
w.r.t.\ $\cal C$ and $\cal O$ is the right semantic equivalence for
that application. However, when the choice of $\cal C$ and $\cal O$ is
not clear, or when proving results that may be re-used in future
applications that may call for extending $\cal C$ or $\cal O$, it is
better to err on the side of caution, and use equivalences that satisfy
(\ref{respect}) and (\ref{congruence}) but need not be fully abstract;
instead the finest equivalence $\equiv_{\it fine}$ for which a
result $p \equiv_{\it fine} q$ can be proved is often preferable,
because this immediately entails that $p \equiv q$ for any coarser
equivalence relation $\equiv$, in particular for an $\equiv$ that may
turn out to be fully abstract for some future choice of $\cal C$ and
$\cal O$. It is for this reason that much actual verification work
employs the finest equivalences that lend themselves for verification
purposes, such as the various variants of bisimulation equivalence
\cite{Mi90ccs}, see e.g.\ \cite{AG08}.
Nevertheless, this paper is devoted to the characterisation of fully
abstract equivalences, and preorders, for a few suitable choices of
$\cal C$ and $\cal O$.

The programme for refinement preorders proceeds along the same lines,
but here it is important to distinguish between good and bad
properties of processes. The counterpart of (\ref{respect}) is
\begin{equation}\label{respect preo}
p \sqsubseteq q ~~\Rightarrow~~ \forall \phi\in {\cal G}.\; (p \models \phi
\Rightarrow q \models \phi)
\end{equation}
where $\cal G$ is the set of \emph{good} properties within $\cal C$,
those that for some applications may be required of a process. If this
holds, $\sqsubseteq$ \emph{respects} or \emph{preserves} the properties
in $\cal G$. When dealing with \emph{bad} properties, those that in
some applications should be avoided, the implication between $p\models
\phi$ and $q\models \phi$ is oriented in the other direction. Since
every bad property $\phi$ can be reformulated as a good property $\neg
\phi$, there is no specific need add a variant of (\ref{respect preo})
for the bad properties.  The counterpart of (\ref{congruence}) is simply
\begin{equation}\label{precongruence}
p \sqsubseteq q ~~\Rightarrow~~ C[p] \sqsubseteq C[q].
\end{equation}
and a preorder $\sqsubseteq$ that satisfies this requirement is called
a \emph{precongruence} for $\cal O$. Now the preorder that is \emph{fully
abstract} w.r.t.\ $\cal G$ and $\cal O$ always exists, and is characterised by
$$p \sqsubseteq q ~~\Leftrightarrow~~ \forall {\cal O}\mbox{-context } C[\_\!\_].\;
\forall\phi\in {\cal G}.\;  (C[p] \models \phi \Rightarrow C[q] \models \phi).$$
It is the coarsest precongruence for $\cal O$ that respects the
properties in $\cal G$.
A characterisation of the preorder $\sqsubseteq$ that is fully
abstract w.r.t.\ a certain $\cal G$ and $\cal O$ automatically yields
a characterisation of the equivalence $\equiv$ that is fully
abstract w.r.t.\ $\cal G$ and $\cal O$, as one has
$p \equiv q$ iff both $p \sqsubseteq q \wedge q \sqsubseteq p$.

In this paper I will propose three main candidates for the set $\cal
G$ of good properties: \emph{safety properties} in \sect{safety},
\emph{liveness properties} in \sect{liveness} and \emph{conditional
  liveness properties} in \sect{conditional}. For the sake of
theoretical completeness I moreover address general \emph{linear time
  properties} in \sect{LT}.

In \sect{LTS} I will define my model of labelled transition systems and
propose a class of $\cal C$ of operators that appear useful in
applications to combine processes. My favourite selection contains
\begin{itemise}
\item the \emph{partially synchronous interleaving operator} of CSP \cite{Ros97},
\item \emph{abstraction} or \emph{concealment} \cite{BK85,Ros97}
\item and the \emph{state operator} \cite{BB88},
\end{itemise}
or any other basis that is equally expressive. With each of the four
choices for $\cal G$ this set of operators determines a fully abstract
preorder, which will be characterised in
Sections~\ref{sec-safety},~\ref{sec-liveness},~\ref{sec-conditional}
and~\ref{sec-LT}.  It turns out that the resulting preorders are
somewhat robust under the precise choice of operators for which one
imposes a precongruence requirement: the same ones are obtained
already without using concealment, and using merely injective renaming
instead of the more general state operator.
In the other direction, I could just as well have used all operators of CSP\@.

\section{Labelled Transition Systems and a Selection of Composition Operators}
\label{sec-LTS}

Let $\Sigma^*$ denote the set of finite sequences over a given set
$\Sigma$, and $\Sigma^\infty$ the set of infinite ones;
$\Sigma^\omega := \Sigma^*\cup\Sigma^\infty$.  Write $\epsilon$ for
the empty sequence, $\sigma\rho$ for the concatenation of sequences
$\sigma\in\Sigma^*$ and $\rho\in\Sigma^\omega$, and $a$ for the
sequence consisting of the single symbol $a \in \Sigma$.  Write
$\sigma\leq\rho$ for ``$\sigma$ is a prefix of~$\rho$'', i.e.\
``$\rho=\sigma \vee \exists \nu\in \Sigma^*.\sigma\nu=\rho$'', and
$\rho<\sigma$ for ``$\sigma\leq\rho$ and $\sigma\neq\rho$''.

I presuppose an infinite action alphabet $\Act$, not containing the
\emph{silent} action $\tau$, and set $\Act_\tau = \Act \cup \{\tau\}$.

\begin{definition}{LTS}
A \emph{labelled transition system} (LTS) is a pair $(\mathbb{P},\rightarrow)$,
where $\mathbb{P}$ is a class of \emph{processes} or \emph{states} and
${\rightarrow} \subseteq \mathbb{P} \times \Act_\tau \times \mathbb{P}$
is a set of \emph{transitions}, such that for each $p \in \mathbb{P}$
and $\alpha \in \Act_\tau$ the class $\{q \in \mathbb{P} \mid
(p,\alpha,q) \in \mathord\rightarrow\}$ is a set.
\end{definition}
Assuming a fixed transition system $(\mathbb{P},\rightarrow)$, I write
$p \goto \alpha q$ for $(p,\alpha,q) \in {\rightarrow}$; this means
that process $p$ can evolve into process $q$, while performing the action $\alpha$.
The ternary relation $\mathord{\dto\_} \subseteq \mathbb{P} \times \Act^*
\times \mathbb{P}$ is the least relation satisfying\\[2pt]
\mbox{}
\hfill
  $p \dto \epsilon p$\enskip,
\hfill
 $\infr{p \goto \tau q}{p \dto \epsilon q}$\enskip,
\hfill
 $\infr{p \goto a q,~ a \not= \tau}{p \dto {a} q}$
\hfill and \hfill
 $\infr{p \dto \sigma q \dto \rho r}{p \dto {\sigma\rho} r}$
\enskip.
\hfill\mbox{}\\[2pt]
This enables a formalisation of the concepts of traces
and complete traces from the introduction.
\begin{definition}{traces}
Let $p\in \mathbb{P}$.
\begin{itemise}
\item $p$ is \emph{deterministic} if, for any $\sigma\in\Act^*$,
  $p\dto{\sigma} q_1$ and $p\dto{\sigma} q_2$ implies that $q_1=q_2$
  and $q_1\gonotto{\tau} r$.
\item $p$ \emph{deadlocks}, notation $p\not\rightarrow$, if there are no
  $\alpha\in\Act_\tau$ and $q\in\mathbb{P}$ such that $p \goto{\alpha} q$.
\item $p$ is \emph{locked} if it can never do a visible action, i.e.\
  if $p \dto{a} q$ for no $a\in\Act$ and $q\in\mathbb{P}$.
\item $p$ \emph{diverges}, notation $p\Uparrow$, if there are
  $p_i\in\mathbb{P}$ for all $i>0$ such that
  $p\goto{\tau}p_1\goto{\tau}p_2\goto{\tau} \cdots$.
\item $a_1 a_2 a_3 \cdots \in \Act^\infty$ is an \emph{infinite trace}
  of $p$ if there are $p_1,p_2,\ldots \in\mathbb{P}$ such that
  $p\dto{a_1}p_1\dto{a_2}p_2\dto{a_3} \cdots$.
\item $\infinite(p)$ denotes the set of infinite traces of $p$.
\item $\ptr(p) := \{\sigma \in \Act^* \mid \exists q.\; p \dto \sigma q\}$
      is the set of \emph{partial traces} of $p$.
\item $\traces(p):=\infinite(p)\cup\ptr(p)$
      is the set of \emph{traces} of $p$.
\item $\deadlocks(p):= \{\sigma \in \Act^* \mid \exists q.\; p \dto
      \sigma q \not\rightarrow\}$ is the set of \emph{deadlock traces} of $p$.
\item $\diverg(p):= \{\sigma \in \Act^* \mid \exists q.\; p \dto
      \sigma q \Uparrow\}$ is the set of \emph{divergence traces} of $p$.
\item $\ct(p):=\infinite(p)\cup\diverg(p)\cup\deadlocks(p)$
      is the set of \emph{complete traces} of $p$.
\end{itemise}
\end{definition}
Note that $\traces(p) = \{\sigma \in \Act^\omega \mid \exists
\rho\in\ct(p).\; \sigma \leq \rho\}$.

To justify that $\ct(p)$ is indeed a correct formalisation of the set
of complete traces of $p$, I postulate that in a neutral environment,
if a process $p\in \mathbb{P}$ has any outgoing transition
$p\goto{\alpha}q$, then within a finite amount of time it will do one
its outgoing transitions. This is called a \emph{progress property};
it says that a process will continue to make progress if possible.

As explained in the introduction, whether a fully abstract equivalence
identifies processes $p$ and $q$ may depend on the existence of a
third process $r$ such that $p\|r$ can be distinguished from $q\|r$.
When restricting attention to a particular labelled transition system
$(\mathbb{P},\rightarrow)$ it might happen that a perfectly reasonable
candidate $r$ happens not to be a member of $\mathbb{P}$, and thus
that the conclusion $p\equiv q$ is arrived at solely as a result
underpopulation of $\mathbb{P}$. To obtain the most robust notions
of equivalence, I therefore assume my LTS to be \emph{universal}, in
the sense that up to isomorphism it contains \emph{any} process one
can imagine.

\begin{definition}{universal}
An LTS $(\mathbb{P},\rightarrow_{\mbox{\bbbs P}})$ is \emph{universal}
if for any other LTS $(\mathbb{Q},\rightarrow_{\mbox{\bbbs Q}})$ there
exists an injective mapping \mbox{$f:\mathbb{Q}\rightarrow\mathbb{P}$},
called an \emph{embedding}, such that, for any $q\in\mathbb{Q}$ and
$p'\in\mathbb{P}$ one has
$f(q)\goto{\alpha}_{\mbox{\bbbs P}}p'$ iff $p'\in\mathbb{P}$ has the
form $f(q')$ for some $q'\in Q$ with $q \goto{\alpha}_{\mbox{\bbbs Q}} q'$.
\end{definition}
The existence of a universal LTS has been established in \cite{vG01}.
Here one needs $\mathbb{P}$ to be a proper class.
All preorders $\sqsubseteq$ that I consider in this paper are defined on
arbitrary LTSs and have the property that $q\sqsubseteq q'
\Leftrightarrow f(p)\sqsubseteq f(q')$, for any embedding $f$.
This means that they are precongruences for isomorphism, and only take
into account the future behaviour of processes, i.e.\ in determining
whether $p \sqsubseteq q$ transitions leading to $p$ or $q$ play no r\^ole.
Thus, a definition of such a preorder on a universal LTS, implicitly
also defines it on any other LTS\@.

\begin{table}
\begin{center}
$\begin{array}{ccc}
\infr{p \goto{\alpha} p'}{p\|^{}_S q \goto{\alpha} p'\|^{}_S q}
\;{\scriptstyle (\alpha\not\in S)} &
\infr{q \goto{\alpha} q'}{p\|^{}_S q \goto{\alpha} p\|^{}_S q'}
\;{\scriptstyle (\alpha\not\in S)} &
\infr{p \goto{a} p' \qquad q \goto{a} q'}{p\|^{}_S q \goto{a} p\|^{}_S q'}
\;{\scriptstyle (a\in S)}
\end{array}$\\[1em]
$\begin{array}{cc}
\infr{p \goto{\alpha} p'}{\tau^{}_I(p) \goto{\alpha} \tau^{}_I(p')}
\;{\scriptstyle (\alpha\not\in I)} &
\infr{p \goto{a} p'}{\tau^{}_I(p) \goto{\tau} \tau^{}_I(p')}
\;{\scriptstyle (a\in I)}
\\[2em]
\infr{p \goto{\tau} p'}{\lambda^m_s(p) \goto{\tau} \lambda^m_{s}(p')} &
\infr{p \goto{a} p'}{\lambda^m_s(p)
  -\!\!\!-\!\!\!\!\goto{\!\!\!\!\!\!\!\!\!\!\!\!\!\!\!a(m,s)}
  \lambda^m_{s(m,a)}(p')}
\end{array}$
\caption{\it Partially synchronous interleaving, abstraction, and the
  state operator}
\label{sos}
\end{center}
\end{table}

I will now do a proposal for the set $\cal O$ that will be my default
choice in this paper. It consists of three operators for combining
processes that appear useful in practical applications.

The first is the \emph{partially synchronous interleaving operator} of
CSP \cite{Ros97}.  It is parametrised with a set $S\subseteq\Act$ of
visible actions on which it synchronises: the composition $p\|^{}_S q$
can perform an action from $S$ only when both $p$ and $q$ perform
it. All other actions from $p$ and $q$ are interleaved: whenever one
of the two components can perform such an action, so can the
composition, while the other component doesn't change its state.
Formally, for any choice of $S\subseteq \Act$,~
$\|^{}_S:\mathbb{P}\times\mathbb{P}\rightarrow\mathbb{P}$ is a binary
operator on \mathbb{P} such that a process $p\|^{}_S q$ can make an
$\alpha$-transition iff this can be inferred by the first three rules
of Table~\ref{sos} from the transitions that $p$ and $q$ can make.
Here $a$ ranges over $\Act$ and $\alpha$ over $\Act_\tau$.

A context $\_\!\_\|^{}_S r$ is widely regarded as a plausible way of
modelling an environment that partially synchronises with processes
under investigation. It is for this reason I include it in $\cal O$.
This argument does not hold for many other process algebraic
operators, such as the choice operator + of CCS \cite{Mi90ccs}. This
is an example of an operator that is useful for \emph{describing}
particular processes, but a context $\_\!\_ + r$ does not really
model a reasonable environment in which one wants to run processes
under investigation. For reasons of algebraic convenience, being a
precongruence for the + is an optional desideratum of refinement
preorders, but it is not such an overriding requirement as being a
precongruence for $\|^{}_S$.

The second operator nominated for membership of $\cal O$ is the unary
\emph{abstraction operator} $\tau^{}_I$ of ACP$_\tau$ \cite{BK85}, also known
as the \emph{concealment} operator of CSP \cite{BHR84,Ros97}. This
operator models a change in the level of abstraction at which
processes are regarded, by reclassifying visible actions as hidden
ones. It is parametrised with the set $I\subseteq\Act$ of visible
actions that one chooses to abstract from, and formally defined by the
next two rules of Table~\ref{sos}.  Abstraction from internal actions
by such a mechanism is an essential part of most work on verification
in a process algebraic setting, and a context $\tau^{}_I(\_\!\_)$
represents a reasonable environment in which to evaluate processes.

My final nominee for the set $\cal O$ of useful composition operators
is the \emph{state operator} $\lambda^m_s$ of \cite{BB88}.  This unary
operator formalises an interface between a process and its environment
that is able to rename actions: if its argument process performs an
action $b$, the interface $\lambda^m_s(\_\!\_)$ may pass on this
action to the environment as $c$, thereby opening up the possibility
of synchronisation with another occurrence of $c$ when using a
composed context $\lambda^m_s(\_\!\_ )\|^{}_S r$. Furthermore, the
interface may remember the actions that have been performed to far,
and make its renaming behaviour dependent on this history. For
instance, if its argument $p$ performs two $a$-actions in a row,
$\lambda^m_s(p)$ may pass these on to the environment as $a_1$ and
$a_2$, respectively.

The state operator $\lambda^m_s$ is parametrised with an
\emph{interface specification} $m=({\cal S},\act,\eff)$, consisting of
set $\cal S$ of \emph{internal states}, and functions $\act:{\cal
S}\times\Act \rightarrow \Act$ and $\eff:{\cal S}\times\Act
\rightarrow \cal S$, as well as a \emph{current state} $s\in\cal S$.
Here $\act$ is a function that renames actions performed by an
argument process $p$ into actions performed by the interface
$\lambda^m_s(p)$; the renaming depends on the internal state of the
state operator, and thus is of type ${\cal S}\times\Act \rightarrow
\Act$. $\eff$ specifies the transformation of one internal state of
the state operator into another, as triggered by the the encounter of
an action of its argument process; it thus is of type ${\cal
S}\times\Act \rightarrow \cal S$. Traditionally, one writes $a(m,s)$
for $\act(s,a)$ and $s(m,a)$ for $\eff(s,a)$. So $a(m,s)$ denotes the
action $a$, as modified by the interface $m$ in state $s$, whereas
$s(m,a)$ denotes the internal state $s$, as modified by the occurrence
of action $a$ of the argument process within the scope of the
interface $m$.  With this notation, the formal definition of the state
operator is given by the last two rules of Table~\ref{sos}.

The special case of a state operator with a singleton set of internal
states is known as a \emph{renaming operator}.  Renaming operators
occur in the languages CCS \cite{Mi90ccs} and CSP \cite{BHR84,Ros97}.
Here I denote a renaming operator as $\lambda^m$, where the redundant
subscript $s$ is omitted, and $m$ trivialises to a function
\mbox{$\act:\Act\rightarrow\Act$}. I speak of an \emph{injective}
renaming operator if $\act(a)=\act(b)$ implies $a=b$. For any
injective renaming operator $\lambda^m$ there exists an inverse
renaming operator $\lambda^{-m}$ (not necessarily injective) such that
for all $p\in\mathbb{P}$, the process $\lambda^{-m}(\lambda^m(p))$
behaves exactly the same as $p$---they are equivalent under all
notions of equivalence considered in this paper.

\section{Safety Properties}\label{sec-safety}

A \emph{safety property} \cite{Lam77} is a property that says that
\begin{quote}\it
something bad will never happen.
\end{quote}
To formulate a canonical safety property, assume that my alphabet
of visible actions contains one specific action $b$, whose occurrence
is \emph{bad}. The canonical safety property now says that
\textbf{$b$ will never happen}.
\begin{definition}{canonical safety}
A process $p$ satisfies the \emph{canonical safety property}, notation
$p\models \textit{safety}(b)$, if no trace of $p$ contains the action $b$.
\end{definition}
To arrive at a general concept of safety property for labelled
transition systems, assume that some notion of \emph{bad} is defined.
Now, to judge whether a process $p$ satisfies this safety property,
one should judge whether $p$ can reach a state in which one would say
that something bad had happened. But all observable behaviour of $p$
that is recorded in a labelled transition system until one comes to such
a verdict, is the sequence of visible actions performed until that
point. Thus the safety property is completely determined by the set
sequences of visible actions that, when performed by $p$, lead to
such a judgement.  Therefore one can just as well define the concept of
a safety property in terms of such a set.
\begin{definition}{safety}
A \emph{safety property} of processes in an LTS is given by a set
$B\subseteq \Act^*$.
A process $p$ \emph{satisfies} this safety property, notation
$p\models \textit{safety}(B)$, when $\ptr(p)\cap B=\emptyset$.
\end{definition}
This formalisation of safety properties is essentially the same as the
one in \cite{AS85} and all subsequent work on safety properties; the
only, non-essential, difference is that I work with transition systems
in which the transitions are labelled, whereas \cite{AS85} and most
related work deals with state-labelled transition systems.


A property is called \emph{trivial} if it either always holds or
always fails. Trivial properties are respected by any equivalence.
The sets $B:=\emptyset$ and $B:=\{\epsilon\}$ specify trivial safety
properties.

\begin{theorem}{canonical safety}
A precongruence for the state operator respects every safety property iff
it respects the canonical safety property.
\end{theorem}

\begin{proof}
``\emph{Only if}'' follows because the canonical safety property is in
  fact a safety property, namely the one with $B$ being the set of
  those sequences that contain the action $b$.

\noindent
``\emph{If}'':
I use here a state operator that remembers exactly what sequence of
actions has occurred so far. Thus the set of internal states of its
interface specification $m$ is $\Act^*$, and furthermore
$\sigma(m,a):=\sigma a$ for all $\sigma\in \Act^*$ and $a\in\Act$.
Now given a safety property $B\subseteq\Act^*$,
let $b\in\Act$ be the special ``bad'' action, and $d\in\Act$ be a
different ``neutral'' action. Define
$a(m,\sigma):=\left\{\begin{array}{ll}b&\mbox{if}~\sigma a \in B\\
d&\mbox{otherwise.}\end{array}\right.$\\
Then $\lambda^m_\epsilon(p)\models\textit{safety}(b)$ iff
$p\models \textit{safety}(B)$. Thus, if $p \sqsubseteq q$ and
$p\models \textit{safety}(B)$, then $\lambda^m_\epsilon(p) \sqsubseteq
\lambda^m_\epsilon(q)$ and $\lambda^m_\epsilon(p)\models \textit{safety}(b)$.
Hence $\lambda^m_\epsilon(q)\models \textit{safety}(b)$, so
$q\models \textit{safety}(B)$.
\end{proof}

Being locked (see \df{traces}) is a safely property, namely with $B$
the set of all sequences over $\Act^*$ of length 1. It can be understood
this way by regarding any occurrence of an action as bad.

\begin{theorem}{locked}
A precongruence for abstraction that respects the property of being locked,
respects the canonical safety property.
\end{theorem}

\begin{proof}
Let $\sqsubseteq$ be a precongruence for abstraction that respects the
property of being locked, and suppose that $p\sqsubseteq q$. Let
$I:=\Act\setminus\{b\}$. Then $\tau_I$ is an operator that renames all actions
other than $b$ into $\tau$; thus if a process of the form $\tau_I(r)$ ever
performs a visible action, it must be $b$.  Now $p\models\emph{safety}(b)
\Leftrightarrow\linebreak \tau_I(p)\models\emph{safety}(b) \Leftrightarrow
\tau_I(p)\mbox{ is locked} \Rightarrow \tau_I(q)\mbox{ is locked}
\Leftrightarrow \tau_I(q)\models\emph{safety}(b) \Leftrightarrow
q\models\emph{safety}(b)$.
\end{proof}
By combining Theorems~\ref{thm-canonical safety} and~\ref{thm-locked} one obtains:

\begin{corollary}{locked}
A precongruence for abstraction and for the state operator that respects the
property of being locked, respects all safety properties.
\end{corollary}

\begin{theorem}{safety}
Any precongruence for $\cal O$ that respects a single nontrivial safety
property, respects every safety property.
\end{theorem}

\begin{proof}
Let $\sqsubseteq$ be a precongruence for $\cal O$ that respects
$\emph{safety}(B)$, where $B \subseteq \Act^*$, $B\neq\emptyset$ and
$\epsilon\not\in B$. Let $\sigma\in\Act^*$ and $a\in\Act$ be such that
$\sigma a \in B$, and no prefix $\rho\leq\sigma$ of $\sigma$ is in $B$.
Let $\emph{safety}(a)$ be the canonical safety property, but with $a$
playing the role of $b$. Naturally, \thm{canonical safety} holds for
this renamed canonical safety property as well. Hence it suffices to
show that $\sqsubseteq$ respects the property $\emph{safety}(a)$.
Let $I:=\Act\setminus\{a\}$. Then $\tau^{}_I$ is an operator that
renames all actions other than $a$ into $\tau$; thus if a process of
the form $\tau^{}_I(s)$ ever performs a visible action, it must be $a$.
Let $r_\sigma$ be a process with $\ct(r)=\{\sigma\}$ and $r_{\sigma a}$ be
a process with $\ct(r)=\{\sigma a\}$. Then, for any choice of $s\in\mathbb{P}$,
$(\tau^{}_I(s)\|^{}_\emptyset r_\sigma)\|^{}_\Act r_{\sigma a}$ is a process  all
of whose traces are prefixes of $\sigma a$, with $\sigma a
\in\ptr((\tau^{}_I (s)\|^{}_\emptyset r_\sigma)\|^{}_\Act r_{\sigma a})$
iff $a\in\ptr(\tau^{}_I(s))$, which is the case iff $s
\not\models\emph{safety}(a)$. Suppose $p\sqsubseteq q$.  Then
$(\tau^{}_I(p)\|^{}_\emptyset r_\sigma)\|^{}_\Act r_{\sigma a}
\sqsubseteq(\tau^{}_I(q)\|^{}_\emptyset r_\sigma)\|^{}_\Act r_{\sigma a}$ and
\\[1ex]
\mbox{}\hfill$p\models\emph{safety}(a) \Leftrightarrow (\tau^{}_I(p)\|^{}_\emptyset
r_\sigma)\|^{}_\Act r_{\sigma a} \models \emph{safety}(B) \Rightarrow
(\tau^{}_I(q)\|^{}_\emptyset r_\sigma)\|^{}_\Act r_{\sigma a} \models
\emph{safety}(B) \Leftrightarrow q\models\emph{safety}(a)$.
\end{proof}
Let $\sqsubseteq_{\it safety}$ denote the preorder that is fully
abstract w.r.t.\ the class of safety properties and $\cal O$.
The following, well-known, theorem characterises this preorder as
\emph{reverse partial trace inclusion}.

\begin{theorem}{safety characterisation}
$p \sqsubseteq_{\it safety} q \Leftrightarrow \ptr(p) \supseteq  \ptr(q)$.
\end{theorem}

\begin{proof}
Define reverse partial trace inclusion, $\sqsubseteq_T^{-1}$, by
$p \sqsubseteq_T^{-1} q$ iff $\ptr(p) \supseteq  \ptr(q)$.

\noindent
``$\Leftarrow$'': It suffices to establish that $\sqsubseteq_T^{-1}$
is a precongruence for $\cal O$ that respects all safety properties.

That $\sqsubseteq_T^{-1}$ is a precongruence for $\cal O$ follows
immediately from the following observations:
$$\begin{array}{@{}rcl@{}}
\ptr(p\|^{}_S q) &=& \{\sigma \in \nu\|^{}_S\xi \mid \nu\in\ptr(p)
\wedge \xi\in\ptr(q)\}\\
\ptr(\tau^{}_I(p)) &=& \{\tau^{}_I(\sigma) \mid \sigma \in\ptr(p)\}\\
\ptr(\lambda^m_s(p)) &=& \{\lambda^m_s(\sigma) \mid \sigma \in\ptr(p)\}.
\end{array}$$
Here $\nu\|^{}_S\xi$ denotes the set of sequences of actions for which
is it possible to mark each action occurrence as \emph{left},
\emph{right} or both, obeying the restriction that an occurrence of
action $a$ is marked both \emph{left} and \emph{right} iff $a\in S$,
such that the subsequence of all \emph{left}-labelled action
occurrences is $\nu$ and the subsequence of all \emph{right}-labelled
action occurrences is $\xi$. Furthermore, the operators $\tau^{}_I$ and
$\lambda^m_s$ on $\Act^*$ are uniquely determined by\vspace{-1ex}
$$\begin{array}{r@{~=~}l@{\qquad}r@{~=~}l}
\tau^{}_I(\epsilon)&\epsilon
&
\tau^{}_I(a\sigma) & \left\{\begin{array}{ll}\tau^{}_I(\sigma) &
\mbox{if}~a\in I\\
a\tau^{}_I(\sigma) & \mbox{otherwise}\end{array}\right.
\\
\lambda^m_s(\epsilon) & \epsilon
&
\lambda^m_s(a\sigma) & a(m,s)\lambda^m_{s(m,a)}(\sigma).
\end{array}$$

To show that $\sqsubseteq_T^{-1}$ respects all safety properties, let
$B\subseteq A^*$, $p \sqsubseteq_T^{-1} q$, and suppose $p \models
\emph{safety}(B)$. Then $\ptr(q) \subseteq \ptr(p)$ and $\ptr(p)\cap B
= \emptyset$. Thus $\ptr(q)\cap B = \emptyset$, i.e.\
$q \models \emph{safety}(B)$, which had to be shown.

``$\Rightarrow$'':
Let $\sqsubseteq$ be any precongruence for $\cal O$ that respects all
safety properties, and suppose $p\sqsubseteq q$.\linebreak[2]
I have to establish that $p\sqsubseteq_T^{-1} q$.
Let $B:=\Act^*\setminus\ptr(p)$. Then $p \models \emph{safety}(B)$.
Thus $q \models \emph{safety}(B)$, i.e.\
$\ptr(q)\cap(\Act^*\setminus\ptr(p))=\emptyset$.
This yields $\ptr(q)\subseteq\ptr(p)$.
\end{proof}
The above characterisation as reverse partial trace inclusion of the
coarsest congruence for $\cal O$ that respects all safety properties,
is rather robust under the choice of $\cal O$. It holds already for
the empty class of operators, and it remains true when adding in all
operators of CSP \cite{BHR84}, CCS \cite{Mi90ccs} or ACP$_\tau$
\cite{BK85}, as $\sqsubseteq_T^{-1}$ is known to be a precongruence
for all of them.

By \thm{safety}, the characterisation also remains valid when
requiring respect for one arbitrary safety property only, instead of
all of them, but to this end all three operators of $\cal O$ are
needed. If we just retain the state operator, by \thm{canonical
safety} it suffices to require respect for the canonical safety
property only.

\section{Liveness Properties}\label{sec-liveness}

A \emph{liveness property} \cite{Lam77} is a property that says that
\begin{quote}\it
something good will eventually happen.
\end{quote}
To formulate a canonical liveness property, assume that the alphabet $\Act$
contains one specific action $g$, whose occurrence
is \emph{good}. The canonical liveness property now says that
\textbf{$g$ will eventually happen}.
\begin{definition}{canonical liveness}
A process $p$ satisfies the \emph{canonical liveness property}, notation
$p\models \textit{liveness}(g)$, if every
complete trace of $p$ contains the action $g$.
\end{definition}
To arrive at a general concept of liveness property for labelled
transition systems, assume that some notion of \emph{good} is defined.
Now, to judge whether a process $p$ satisfies this liveness property,
one should judge whether $p$ can reach a state in which one would say
that something good had happened. But all observable behaviour of $p$
that is recorded in a labelled transition system until one comes to such
a verdict, is the sequence of visible actions performed until that
point. Thus the liveness property is completely determined by the set
sequences of visible actions that, when performed by $p$, lead to
such a judgement.  Therefore one can just as well define the concept of
a liveness property in terms of such a set.
\begin{definition}{liveness}
A \emph{liveness property} of processes in an LTS is given by a set
$G\subseteq \Act^*$.
A process $p$ \emph{satisfies} this liveness property, notation
$p\models \textit{liveness}(G)$, when each complete trace of $p$ has a
prefix in $G$.
\end{definition}
This formalisation of liveness properties is essentially different from the
one in \cite{AS85} and most subsequent work on liveness properties;
this point is discussed in \sect{LT}.

Just as for safety properties,
the sets $G:=\emptyset$ and
$G:=\{\epsilon\}$ specify trivial liveness properties.

\begin{theorem}{canonical liveness}
A precongruence for the state operator respects every liveness
property iff it respects the canonical liveness property.
\end{theorem}

\begin{proof}
Just like the proof of \thm{canonical safety}.
\end{proof}

A process $p$ has the \emph{initial progress} property if it cannot
immediately diverge or deadlock, i.e.\ if
$\epsilon\not\in\diverg(p)\cup\deadlocks(p)$.  This is a liveness
property, namely with $G$ the set of all sequences over $\Act^*$ of
length 1. It can be understood this way by regarding any occurrence of
an action as good.

\begin{theorem}{initial progress}
A precongruence for abstraction that respects the initial progress
property, respects the canonical liveness property.
\end{theorem}

\begin{proof}
Just like the proof of \thm{locked}.
\end{proof}
By combining Theorems~\ref{thm-canonical liveness}
and~\ref{thm-initial progress} one obtains:

\begin{corollary}{initial progress}
A precongruence for abstraction and for the state operator that respects the
initial progress property, respects all liveness properties.
\end{corollary}

\begin{conjecture}{liveness}
Any precongruence for $\cal O$ that respects a single nontrivial liveness
property, respects every liveness property.
\end{conjecture}
%
%
Let $\sqsubseteq_{\it liveness}$ denote the preorder that is fully
abstract w.r.t.\ the class of liveness properties and $\cal O$.  I
will proceed to characterise $\sqsubseteq_{\it liveness}$ as the
preorder $\sqsubseteq_{\it FDI}^\bot$ based on failures, divergences and
infinite traces that is also used in the work on CSP
\cite{Ros97}. \emph{Failures} of a process $p$ are defined below; they
are pairs $\langle \sigma, X\rangle$ such that $p$ can perform the
sequence of visible actions $\sigma$ and then reach a state in which
no further progress can be made in case the environment allows only
those visible actions to occur that are listed in $X$. The preorder
$\sqsubseteq_{\it FDI}^\bot$ does not take into account any information
about the behaviour of processes that can be thought of taking place
after a divergence. One of the ways to erase this information from the
set of failures, divergences and infinite traces of a process is by
means of \emph{flooding}. Flooded sets of failures, divergences and
infinite traces are indicated by the subscript $_\bot$.

\begin{definition}{FDI}
Let $p\in \mathbb{P}$.
\begin{itemise}
\item $\ini(p):=\{\alpha\in\Act_\tau \mid \exists q.\; p \goto{\alpha} q\}$.
\item $\failures(p):=\{\langle \sigma, X\rangle \in \Act^*\times\pow(\Act) \mid
  \exists q.\; p \dto{\sigma} q \wedge \ini(q) \cap (X\cup\{\tau\})=\emptyset\}$.
\item $\divd(p):= \{\sigma\rho \mid \sigma\in\diverg(p) \wedge \rho\in\Act^*\}$.
\item $\infd(p):= \{\sigma\rho \mid \sigma\in\diverg(p) \wedge \rho\in\Act^\infty\}$.
\item $\faild(p):= \{\langle\sigma\rho,X\rangle \mid
  \sigma\in\diverg(p) \wedge \rho\in\Act^* \wedge X\subseteq \Act\}$.
\end{itemise}
\end{definition}
So $\deadlocks(p)\mathbin=\{\sigma | \langle\sigma,A\rangle\mathbin\in\failures(p)\}$ and
$\ptr(p)\mathbin=\diverg(p) \cup\{\sigma | \langle\sigma,\emptyset\rangle\mathbin\in\failures(p)\}$.

\begin{theorem}{liveness characterisation}
$p \sqsubseteq_{\it liveness} q ~~\Leftrightarrow~~
\begin{array}[t]{@{}r@{~\supseteq~}l@{}}
\divd(p) & \divd(q) \wedge \mbox{}\\
\infd(p) & \infd(q) \wedge \mbox{}\\
\faild(p) & \faild(q).
\end{array}$
\end{theorem}

\begin{proof}
Let $\sqsubseteq_{\it FDI}^\bot$ be the preorder defined by:
$p\sqsubseteq_{\it FDI}^\bot q$ iff the right-hand side of \thm{liveness
  characterisation} holds.

\noindent
``$\Leftarrow$'':
It suffices to establish that
$\sqsubseteq_{\it FDI}^\bot$ is a liveness respecting precongruence.

To show that $\sqsubseteq_{\it FDI}^\bot$ respects liveness, let
$G\subseteq \Act^*$, $p\sqsubseteq_{\it FDI}^\bot q$, and suppose
$p \models \textit{liveness}(G)$. I need to show that
$q \models \textit{liveness}(G)$. So suppose $\sigma \in \ct(q)$.  Then
one out of three possibilities must apply: either
$\sigma\in\diverg(g) \subseteq\divd(g) \subseteq\divd(p)$ or
$\sigma \in \infinite(q) \subseteq \infd(q) \subseteq \infd(p)$ or
$\langle\sigma,\Act\rangle \in \failures(q) \subseteq \faild(q)
\subseteq \faild(q)$.
In the first case $\rho\in\diverg(p)\subseteq \ct(p)$
for some $\rho \leq \sigma$; in the second case either $\sigma
\in\infinite(p) \subseteq \ct(p)$ or $\rho\in\diverg(p)\subseteq \ct(p)$
for some $\rho < \sigma$; and in the third case either
$\langle\sigma,\Act\rangle\in\failures(p)$ or
$\rho\in\diverg(p)\subseteq \ct(p)$ for some $\rho \leq \sigma$.
In all three cases $\rho\in \ct(p)$ for some $\rho\leq\sigma$.
Since $p \models \textit{liveness}(G)$, there must be a $\nu\leq\rho$
with $\nu\in G$. As $\nu\leq\sigma$ it follows that
$q \models \textit{liveness}(G)$.

That $\sqsubseteq_{\it FDI}^\bot$ is a precongruence for $\|^{}_S$ and $\tau^{}_I$
has been established in \cite{Ros97} by means of the following observations:
$$\begin{array}{@{}rcl@{}}
\divd(p\|^{}_S q) &=& \{\sigma\rho \mid \exists
  \langle\nu,X\rangle\in\faild(p), \xi \in \divd(q).\; \sigma\in
  \nu\|^{}_S\xi \wedge \rho\in\Act^*\} \cup \mbox{}\\
 && \{\sigma\rho \mid \exists \nu\in\divd(p),
  \langle\xi,X\rangle\in\faild(q).\; \sigma\in \nu\|^{}_S\xi\wedge \rho\in\Act^*\}\\
\infd(p\|^{}_S q) &=&
  \{\sigma \mid \exists  \nu\in\infd(p), \xi \in \infd(q).\;
  \sigma\in \nu\|^{}_S\xi\} \cup \mbox{}\\
 && \{\sigma \mid \exists \langle\nu,X\rangle\in\faild(p), \xi \in \infd(q).\;
  \sigma\in \nu\|^{}_S\xi\} \cup \mbox{}\\
 && \{\sigma \mid \exists \nu\in\infd(p),
  \langle\xi,X\rangle\in\faild(q).\; \sigma\in \nu\|^{}_S\xi\} \cup \mbox{}\\
 && \{\sigma\rho \mid \sigma \in \divd(p\|^{}_S q) \wedge \rho\in\Act^\infty\}\\
\faild(p\|^{}_S q) &=& \{\langle\sigma,X\cup Y\rangle \mid \exists
  \langle\nu,X\rangle\in\faild(p), \langle\xi,Y\rangle\in\faild(q).\\
 && \hfill X\setminus S = Y\setminus S \wedge
    \sigma\in \nu\|^{}_S\xi\} \cup \mbox{}\\
 && \{\langle\sigma,X\rangle \mid \sigma\in\divd(p\|^{}_S q)
  \wedge X\subseteq\Act\}.\\
\divd(\tau^{}_I(p)) &=& \{\tau^{}_I(\sigma)\rho \mid \tau^{}_I(\sigma), \rho \in A^*
  \wedge \sigma \in\infd(p)\cup\divd(p)\} \\
\infd(\tau^{}_I(p)) &=& \{\tau^{}_I(\sigma) \mid \tau^{}_I(\sigma)\in A^\infty
  \wedge \sigma \in\infd(p)\} \\
 && \cup \{\sigma\rho \mid \sigma\in\divd(\tau^{}_I(p)) \wedge \rho\in\Act^\infty\}\\
\faild(\tau^{}_I(p)) &=& \{\langle\tau^{}_I(\sigma),X\rangle \mid
  \langle\sigma,X\cup I\rangle \in \faild(p)\} \cup \mbox{}\\
 && \{\langle\sigma,X\rangle \mid \sigma\in\divd(\tau^{}_I(p))
  \wedge X\subseteq\Act\}.
\end{array}$$
Here $\tau^{}_I(\sigma)$ for $\sigma\in\Act^\infty$ is the supremum,
w.r.t.\ the prefix order $\leq$ on $\Act^\omega$, of the set
$\{\tau^{}_I(\rho) \mid \rho < \sigma\}$.

Likewise,  $\sqsubseteq_{\it FDI}^\bot$ is a congruence for $\lambda^m_s$:
$$\begin{array}{@{}rcl@{}}
\divd(\lambda^m_s(p)) &=& \{\lambda^m_s(\sigma)\rho \mid \sigma\in\divd(p) \wedge \rho\in\Act^*\} \\
\infd(\lambda^m_s(p)) &=& \{\lambda^m_s(\sigma) \mid \sigma\in\infd(p)\}
  \cup \{\sigma\rho \mid \sigma\in\divd(\lambda^m_s(p)) \wedge \rho\in\Act^\infty\}\\
\faild(\lambda^m_s(p)) &=& \{\langle\lambda^m_s(\sigma),X\rangle \mid
                \langle\sigma,\lambda^{-m}_s(X)\rangle \in\faild(p)\} \cup \mbox{}\\
 && \{\langle\sigma,X\rangle \mid \sigma\in\divd(\lambda^m_s(p)) \wedge X\subseteq\Act\}.\\
\end{array}$$
Here $\lambda^{-m}_s(X) := \{a \in A \mid a(m,s)\in X\}$.\pagebreak[3]

``$\Rightarrow$'': Let $\sqsubseteq$ be any liveness respecting
precongruence, and suppose $p\sqsubseteq q$. I have to establish that
$p\sqsubseteq_{\it FDI}^\bot q$.  W.l.o.g.\ I may assume that neither $p$
nor $q$ has any trace containing the action $g$. For let $\lambda^m$
be an injective renaming operator such that $g$ is not in the image of
$\lambda^m$.  Then $\lambda^m(p) \sqsubseteq \lambda^m(q)$. Suppose
one can establish $\lambda^m(p) \sqsubseteq_{\it FDI}^\bot \lambda^m(q)$.
Since $\sqsubseteq_{\it FDI}^\bot$ is a precongruence for renaming, this
yields $p\equiv_{\it FDI}^\bot \lambda^{-m}(\lambda^m(p)) \sqsubseteq_{\it
FDI}^\bot \lambda^{-m}(\lambda^m(p)) \equiv_{\it FDI}^\bot q$.

Suppose $\divd(p) \not\supseteq \divd(q)$; say $\sigma\in
\divd(q)\setminus \divd(p)$.  So there is no $\rho\leq \sigma$ with
$\rho\in\diverg(p)$.  Let $r$ be a deterministic process such that
$\ct(r)=\{\rho g \mid \rho \leq \sigma\}$. Then each complete trace of
$p\|^g r$ contains $g$.  Here I write $\|^g$ for
\plat{$\|^{}_{A\setminus\{g\}}$}, the interleaving operator that
synchronises on all visible actions except $g$.  As $\sqsubseteq$ is a
precongruence, $p\sqsubseteq q$ implies $p\|^g r \sqsubseteq q\|^g r$,
and since $\sqsubseteq$ respects the canonical liveness property, I
obtain that each complete trace of $q\|^g r$ must contain
$g$. However, $\rho\mathbin\in\diverg(q)$ for some
$\rho\mathbin\leq\sigma$. So $\rho \mathbin\in \ct(q\|^g r)$, although
$\rho$ does not contain~$g$.

Suppose $\infd(p) \not\supseteq \infd(q)$; say $\sigma\in
\infd(q)\setminus \infd(p)$.  So $\sigma \not\in \infinite(p)$ and
there is no $\rho < \sigma$ with $\rho\in\diverg(p)$.  Let $r$ be a
deterministic process such that $\ct(r)=\{\rho g \mid \rho < \sigma\}
\cup \{\sigma\}$. Then each complete trace of $p\|^g r$ contains $g$.
As $\sqsubseteq$ is a precongruence, $p\sqsubseteq q$ implies $p\|^g r
\sqsubseteq q\|^g r$, and since $\sqsubseteq$ respects the canonical
liveness property, I obtain that each complete trace of $q\|^g r$
must contain~$g$. However, either $\sigma\in\infinite(q)$ or
$\rho\in\diverg(q)$ for some $\rho<\sigma$. So either $\sigma \in
\ct(q\|^g r)$ or $\rho \in \ct(q\|^g r)$, and neither $\sigma$ nor
$\rho$ contains $g$.

Suppose $\faild(p) \not\supseteq \faild(q)$; say
$\langle\sigma,X\rangle\in \faild(q)\setminus \faild(p)$.  So
$\langle\sigma,X\rangle\not\in\failures(p)$ and there is no $\rho\leq
\sigma$ with $\rho\in\diverg(p)$. Let $r$ be a deterministic process
with $\ct(r)=\{\rho g \mid \rho < \sigma\} \cup \{\sigma a \mid a \in
X\}$, and consider the liveness property given by $G:=\{\rho g \mid
\rho<\sigma\} \cup \{\sigma a \mid a \in X\}$. Then $p\|^g r \models
\textit{liveness}(G)$.  As $\sqsubseteq$ is a precongruence,
$p\sqsubseteq q$ implies $p\|^g r \sqsubseteq q\|^g r$, and since
$\sqsubseteq$ respects liveness properties, also $q\|^g r \models
\textit{liveness}(G)$.  However, either
$\langle\sigma,X\rangle\in\failures(q)$ or there is an $\rho\leq
\sigma$ with $\rho\in\diverg(q)$. So either $\sigma \in \ct(q\|^g r)$
or $\rho \in \ct(q\|^g r)$ for some $\rho\leq \sigma$, contradicting
that $q\|^g r \models \textit{liveness}(G)$.
\end{proof}
The standard refinement preorder used in CSP is in fact the
\emph{failures-divergences} preorder $\sqsubseteq_{\it FD}$, defined
exactly like $\sqsubseteq_{\it FDI}^\bot$, but abstracting from the
infinite traces. As remarked in \cite{Ros97}, this can be done because
in CSP one normally restricts attention to processes $p$ with the property
that for any $\sigma\in\Act^*$ either $\sigma\in\divd(p)$ or there are
only finitely many processes $q$ with \plat{$p \dto{\sigma} q$}. For
such processes the set $\infd(p)$ is, with K\"onigs Lemma, completely
determined by $\faild(p)$ and $\divd(p)$, and thus need not be
explicitly recorded. When extending CSP to processes not having this
property, the component $\infd$ should be added to the semantics of
processes \cite{Ros97}. In fact, $\sqsubseteq_{\it FDI}^\bot$ is the
coarsest precongruence for $\cal O$ contained in $\sqsubseteq_{\it
FD}$: if $p$, $q$ and $r$ are the processes used in the $\infd$-case
of the above proof, and $I:=\Act\setminus \{g\}$, then $\epsilon \in
\divd(\tau^{}_I(q\|^g r))\setminus\divd(\tau^{}_I(p\|^g r))$.

The above characterisation as $\sqsubseteq_{\it FDI}^\bot$ of the coarsest
congruence for $\cal O$ that respects all liveness properties, is
somewhat robust under the choice of $\cal O$. It holds already for
with just $\|^g$ and injective renaming (for these are the only two
operators that are just in the proof), and it remains true when
adding in all operators of CSP \cite{BHR84}, as $\sqsubseteq_{\it
FDI}^\bot$ is known to be a precongruence for all of them \cite{Ros97}.

By \cor{initial progress}, the above characterisation also remains
valid when requiring respect for the initial progress property only,
but to this end all three operators of $\cal O$ are needed.
%
This result has in essence been obtained already by Bill Roscoe in
\cite{Ros97}. The state operator does not feature in \cite{Ros97}; its
r\^ole in this full abstraction result is taken over by a renaming
operator that allows renaming an action $a$ into a choice between two
actions $b$ and $c$.  When ignoring this difference in syntax,
\thm{liveness characterisation} can be obtained as an immediate
corollary of \cor{initial progress} and that result. The main reason
for using the above proof instead is to show that the concealment or
abstraction operator is not needed here.

By \thm{canonical liveness}, $\sqsubseteq_{\it FDI}^\bot$ is even fully
abstract w.r.t.\ the partially synchronous interleaving and state
operators, and the canonical liveness property. This result, like the
full abstraction result of \cite{Ros97}, does not hold without the
state operator, or something equally powerful, even if renaming and
abstraction is allowed to be used. Namely, as pointed out by Antti
Puhakka \cite{Pu01}, one would fail to distinguish the following two
processes:\\
\expandafter\ifx\csname graph\endcsname\relax \csname newbox\endcsname\graph\fi
\expandafter\ifx\csname graphtemp\endcsname\relax \csname newdimen\endcsname\graphtemp\fi
\setbox\graph=\vtop{\vskip 0pt\hbox{%
    \special{pn 8}%
    \special{ar 1050 400 50 50 0 6.28319}%
    \special{pa 1085 365}%
    \special{pa 1200 150}%
    \special{pa 1400 150}%
    \special{pa 1500 400}%
    \special{pa 1350 550}%
    \special{pa 1200 500}%
    \special{pa 1085 435}%
    \special{sp}%
    \special{sh 1.000}%
    \special{pa 1160 506}%
    \special{pa 1085 435}%
    \special{pa 1185 463}%
    \special{pa 1160 506}%
    \special{fp}%
    \graphtemp=.5ex\advance\graphtemp by 0.100in
    \rlap{\kern 1.450in\lower\graphtemp\hbox to 0pt{\hss $\tau$\hss}}%
    \special{ar 550 400 50 50 0 6.28319}%
    \special{pa 600 400}%
    \special{pa 1000 400}%
    \special{fp}%
    \special{sh 1.000}%
    \special{pa 900 375}%
    \special{pa 1000 400}%
    \special{pa 900 425}%
    \special{pa 900 375}%
    \special{fp}%
    \graphtemp=\baselineskip\multiply\graphtemp by -1\divide\graphtemp by 2
    \advance\graphtemp by .5ex\advance\graphtemp by 0.400in
    \rlap{\kern 0.800in\lower\graphtemp\hbox to 0pt{\hss $a$\hss}}%
    \special{ar 50 400 50 50 0 6.28319}%
    \graphtemp=.5ex\advance\graphtemp by 0.400in
    \rlap{\kern 0.050in\lower\graphtemp\hbox to 0pt{\hss $\bullet$\hss}}%
    \special{pa 100 400}%
    \special{pa 500 400}%
    \special{fp}%
    \special{sh 1.000}%
    \special{pa 400 375}%
    \special{pa 500 400}%
    \special{pa 400 425}%
    \special{pa 400 375}%
    \special{fp}%
    \graphtemp=\baselineskip\multiply\graphtemp by -1\divide\graphtemp by 2
    \advance\graphtemp by .5ex\advance\graphtemp by 0.400in
    \rlap{\kern 0.300in\lower\graphtemp\hbox to 0pt{\hss $a$\hss}}%
    \graphtemp=.5ex\advance\graphtemp by 0.400in
    \rlap{\kern 2.050in\lower\graphtemp\hbox to 0pt{\hss $\not\equiv_{\it liveness}$\hss}}%
    \special{ar 3550 400 50 50 0 6.28319}%
    \special{pa 3585 365}%
    \special{pa 3700 150}%
    \special{pa 3900 150}%
    \special{pa 4000 400}%
    \special{pa 3850 550}%
    \special{pa 3700 500}%
    \special{pa 3585 435}%
    \special{sp}%
    \special{sh 1.000}%
    \special{pa 3660 506}%
    \special{pa 3585 435}%
    \special{pa 3685 463}%
    \special{pa 3660 506}%
    \special{fp}%
    \graphtemp=.5ex\advance\graphtemp by 0.100in
    \rlap{\kern 3.950in\lower\graphtemp\hbox to 0pt{\hss $\tau$\hss}}%
    \special{ar 3050 400 50 50 0 6.28319}%
    \special{pa 3100 400}%
    \special{pa 3500 400}%
    \special{fp}%
    \special{sh 1.000}%
    \special{pa 3400 375}%
    \special{pa 3500 400}%
    \special{pa 3400 425}%
    \special{pa 3400 375}%
    \special{fp}%
    \graphtemp=\baselineskip\multiply\graphtemp by -1\divide\graphtemp by 2
    \advance\graphtemp by .5ex\advance\graphtemp by 0.400in
    \rlap{\kern 3.300in\lower\graphtemp\hbox to 0pt{\hss $a$\hss}}%
    \special{ar 2550 400 50 50 0 6.28319}%
    \graphtemp=.5ex\advance\graphtemp by 0.400in
    \rlap{\kern 2.550in\lower\graphtemp\hbox to 0pt{\hss $\bullet$\hss}}%
    \special{pa 2600 400}%
    \special{pa 3000 400}%
    \special{fp}%
    \special{sh 1.000}%
    \special{pa 2900 375}%
    \special{pa 3000 400}%
    \special{pa 2900 425}%
    \special{pa 2900 375}%
    \special{fp}%
    \graphtemp=\baselineskip\multiply\graphtemp by -1\divide\graphtemp by 2
    \advance\graphtemp by .5ex\advance\graphtemp by 0.400in
    \rlap{\kern 2.800in\lower\graphtemp\hbox to 0pt{\hss $a$\hss}}%
    \special{ar 2950 50 50 50 0 6.28319}%
    \special{pa 2588 367}%
    \special{pa 2912 83}%
    \special{fp}%
    \special{sh 1.000}%
    \special{pa 2821 130}%
    \special{pa 2912 83}%
    \special{pa 2854 168}%
    \special{pa 2821 130}%
    \special{fp}%
    \graphtemp=\baselineskip\multiply\graphtemp by -1\divide\graphtemp by 2
    \advance\graphtemp by .5ex\advance\graphtemp by 0.225in
    \rlap{\kern 2.750in\lower\graphtemp\hbox to 0pt{\hss $a$~\hss}}%
    \hbox{\vrule depth0.525in width0pt height 0pt}%
    \kern 3.969in
  }%
}%

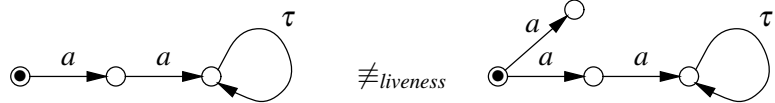
\centerline{\raisebox{1em}{\box\graph}}\vspace{1ex}

\section{Conditional Liveness Properties}\label{sec-conditional}

Figure~\ref{conditional} presents two processes that have the same
\begin{figure}
\input{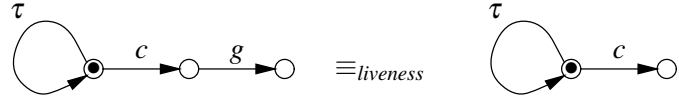}
\centerline{\raisebox{0em}{\box\graph}}
\caption{\it Two processes with the same liveness properties but different
  conditional liveness properties}
\label{conditional}
\end{figure}
liveness properties in any CSP-context. The fact that only the
left-hand process \emph{can} do something good doesn't matter here, as
neither of the two processes is \emph{guaranteed} to do something
good: they may never proceed beyond their initial $\tau$-loops.
Nevertheless, from a practical point of view, the difference between
these two processes may be enormous. It could be that the action $c$
comes with a huge cost, that is only worth making when something good
happens afterwards. Only the right-hand side process is able to incur
the cost without any benefits, and for this reason it lacks an
important property that the left-hand process has. I call such
properties \emph{conditional liveness properties} \cite{GV06,Levy08}.
A \emph{conditional liveness property} is a property that says that
\begin{quote}\it
under certain conditions something good will eventually happen.
\end{quote}
To formulate a canonical conditional liveness property, assume that
the alphabet $\Act$ contains two specific action $c$ and $g$, where
the occurrence of $c$ is the condition, and the occurrence of $g$
is \emph{good}. The canonical conditional liveness property now says that
\textbf{if $c$ occurs then $g$ will eventually happen}.
\begin{definition}{canonical conditional liveness}
A process $p$ satisfies the \emph{canonical conditional liveness property},
notation $p\models \textit{liveness}_c(g)$, if every complete trace of
$p$ that contains the action $c$ also contains the action $g$.
\end{definition}
To arrive at a general concept of conditional liveness property for
labelled transition systems, assume that some condition, and some
notion of \emph{good} is defined.  Now, to judge whether a process $p$
satisfies this conditional liveness property, one should judge first
of all in which states the condition is fulfilled.  All observable
behaviour of $p$ that is recorded in a labelled transition system
until one comes to such a verdict, is the sequence of visible actions
performed until that point. Thus the condition is completely
determined by the set sequences of visible actions that, when
performed by $p$, lead to such a judgement. Next one should judge
whether $p$ can reach a state in which one would say that something
good had happened. Again, this judgement can be expressed in terms of
the sequences of visible actions that lead to such a state.

\begin{definition}{conditional liveness}
A \emph{conditional liveness property} of processes in an LTS is given
by two sets $C,G\subseteq \Act^*$.
A process $p$ \emph{satisfies} this conditional liveness property, notation
$p\models \textit{liveness}_C(G)$, when each complete trace of
$p$ that has a prefix in $C$, also has prefix in $G$.
\end{definition}
For the sake of added generality, one could make the notion of success
dependent on the particular sequence of actions that fulfilled the
condition.  This would make $G$ a function from $C$ to $\pow(\Act^*)$
and the requirement would be that each complete trace of $p$ that has
a prefix $\sigma\in C$, also has prefix in $G(\sigma)$.  However, such
a generalised conditional liveness property can be expressed as a
conjunction of standard ones, and a preorder that respects a
given collection of properties also respects their conjunction.

\begin{theorem}{canonical conditional liveness}
A precongruence for the state operator respects every conditional
liveness property iff it respects the canonical conditional liveness property.
\end{theorem}

\begin{proof}
``\emph{Only if}'' follows because the canonical conditional liveness
  property is in fact a conditional liveness property, namely the one
  with $C$ being the set of those sequences that contain the action
  $c$, and $G$ the set of those sequences that contain the action $g$.

\noindent
``\emph{If}'':  Again I use  a state  operator that  remembers exactly
what sequence of actions has occurred so far. Thus the set of internal
states   of  its   interface  specification   $m$  is   $\Act^*$,  and
$\sigma(m,a):=\sigma  a$ for  all $\sigma\in  \Act^*$  and $a\in\Act$.
Note that the properties $\textit{liveness}_C(G)$ and
$\textit{liveness}_{C\setminus G}(G)$ are satisfied by the same
processes, so w.l.o.g.\ I may restrict attention to properties
$\emph{liveness}_C(G)$ with $C\cap G=\emptyset$.
Given such a property, define
$$a(m,\sigma):=\left\{\begin{array}{ll}c&\mbox{if}~\sigma a \in C\\
g&\mbox{if}~\sigma a \in G\\
d&\mbox{otherwise.}\end{array}\right.$$
Then $\lambda^m_\epsilon(p)\models\textit{liveness}_c(g)$ iff
$p\models \textit{liveness}_C(G)$. Thus, if $p \sqsubseteq q$ and
$p\models \textit{liveness}_C(G)$, then $\lambda^m_\epsilon(p) \sqsubseteq
\lambda^m_\epsilon(q)$ and $\lambda^m_\epsilon(p)\models
\textit{liveness}_c(g)$.
Hence $\lambda^m_\epsilon(q)\models \textit{liveness}_c(g)$, so
$q\models \textit{liveness}_C(G)$.
\end{proof}


An element $\sigma\not\in\diverg(p)\cup\deadlocks(p)$ is called a
\emph{deadlock/divergence trace} of a process $p$. For any
$\sigma\in\Act^*$, not having a deadlock/divergence trace $\sigma$ is
a conditional liveness property, namely with $C:=\{\sigma\}$ and
$G:=\{\sigma a\mid a\in\Act\}$.
Using similar techniques as for \cor{locked}, one can establish:

\begin{corollary}{deadlock/divergence}
A precongruence for abstraction and for the state operator that
respects the property of having no deadlock/divergence trace $c$,
respects all liveness properties.
\end{corollary}
Let $\sqsubseteq_{\it cond.~liveness}$ denote the preorder that is fully
abstract w.r.t.\ the class of conditional liveness properties and $\cal O$.
Furthermore, write $\sqsubseteq_{\it d/d}$ for the coarsest
precongruence for $\cal O$ such that $q \sqsubseteq_{\it d/d} p$
implies $\deadlocks(q)\cup\diverg(q) \subseteq
\deadlocks(p)\cup\diverg(p)$.

\begin{corollary}{dd}
$p \sqsubseteq_{\it cond.~liveness} q$ iff $q \sqsubseteq_{\it d/d} p$.
\end{corollary}

\begin{proof}
``\emph{If}'' follows immediately from \cor{deadlock/divergence}.
``\emph{Only if}'' follows from the observation that the absence of any
deadlock/divergence trace $\sigma$ is a conditional liveness property.
\end{proof}
Antti Puhakka \cite{Pu01} has given a characterisation of the coarsest
congruence that preserves deadlock/\linebreak[2]divergence traces,
$\equiv_{\it d/d}$.  His arguments easily extend to a characterisation
of $\sqsubseteq_{\it d/d}$ and hence, using \cor{dd}, of
$\sqsubseteq_{\it cond.~liveness}$.  Below I will give a direct proof
of the same result. It shows that this characterisation is already
valid when merely requiring the precongruence property for $\|_S$ and
injective renaming.

As for $\sqsubseteq_{\it liveness}$, the characterisation of
$\sqsubseteq_{\it cond.~liveness}$ is in terms of failures,
divergences and infinite traces, and again some information needs to
be erased, but less than in the case of $\sqsubseteq_{\it liveness}$.
This time we need to forget about failures $\langle \sigma,X\rangle
\in\failures(p)$ such that $\sigma\in\diverg(p)$, and about infinite
traces of $p$ that have arbitrary long prefixes in $\diverg(p)$. In
\cite{Pu01} this is achieved by removal of such failures and infinite
traces; here, in order to stress the similarity with the refinement
preorder of CSP, I equivalently apply the method of flooding.

\begin{definition}{Antti}
Let $p\in \mathbb{P}$.
\begin{itemise}
\item $\infdd(p):= \infinite(p) \cup \{\sigma\in\Act^\infty \mid
  \forall\rho<\sigma \exists \nu\in\diverg(p).\; \rho\leq\nu<\sigma\}$.
\item $\faildd(p):= \failures(p) \cup \{\langle\sigma,X\rangle \mid
  \sigma\in\diverg(p)\wedge X\subseteq \Act\}$.
\end{itemise}
\end{definition}

\begin{theorem}{conditional liveness characterisation}
$p \sqsubseteq_{\it cond.~liveness} q ~~\Leftrightarrow~~
\begin{array}[t]{@{}r@{~\supseteq~}l@{}}
\diverg(p) & \diverg(q) \wedge \mbox{}\\
\infdd(p) & \infdd(q) \wedge \mbox{}\\
\faildd(p) & \faildd(q).
\end{array}$
\end{theorem}

\begin{proof}
Let $\sqsubseteq_{\it FDI}^d$ be the preorder defined by:
$p\sqsubseteq_{\it FDI}^d q$ iff the right-hand side of \thm{liveness
  characterisation} holds.

\noindent
``$\Leftarrow$'':
It suffices to establish that
$\sqsubseteq_{\it FDI}^d$ is a precongruence for $\cal O$ that respects
all conditional liveness properties.

To show that $\sqsubseteq_{\it FDI}^d$ respects conditional liveness
properties, let $C,G\subseteq \Act^*$, $p\sqsubseteq_{\it FDI}^d q$, and
suppose $p \models \textit{liveness}_C(G)$. I need to show that
$q \models \textit{liveness}_C(G)$. So suppose $\sigma \in \ct(q)$ and
$\rho\in C$ for some prefix $\rho\leq\sigma$.  Then
one out of three possibilities must apply: either
$\sigma\in\diverg(g) \subseteq\diverg(p)$ or
$\sigma \in \infinite(q) \subseteq \infdd(q) \subseteq \infdd(p)$ or
$\langle\sigma,\Act\rangle \in \failures(q) \subseteq \faildd(q)
\subseteq \faildd(p)$.
In the first and last case, one has $\sigma\in\ct(p)$.
Since $p \models \textit{liveness}_C(G)$, there must be a $\xi\leq\sigma$
with $\xi\in G$, which had to be shown.
In the second case either $\sigma \in\infinite(p) \subseteq \ct(p)$,
in which case the argument proceeds as above, or
$\exists\nu\in\diverg(p) \subseteq \ct(p)$ with $\rho\leq\nu<\sigma$.
In the latter case, there must be a $\xi\leq\nu$ with $\xi\in G$, and
as $\nu<\sigma$ it follows that $q \models \textit{liveness}_C(G)$.

That $\sqsubseteq_{\it FDI}^d$ is a precongruence for $\|^{}_S$,
$\tau^{}_I$ and $\lambda^m_s$ follows from the following observations:
$$\begin{array}{@{}rcl@{}}
\diverg(p\|^{}_S q) &=& \{\sigma \mid \exists
  \langle\nu,X\rangle\in\faildd(p), \xi \in \diverg(q).\; \sigma\in
  \nu\|^{}_S\xi\} \cup \mbox{}\\
 && \{\sigma \mid \exists \nu\in\diverg(p),
  \langle\xi,X\rangle\in\faildd(q).\; \sigma\in \nu\|^{}_S\xi\}\\
\infdd(p\|^{}_S q) &=&
  \{\sigma \mid \exists  \nu\in\infdd(p), \xi \in \infdd(q).\;
  \sigma\in \nu\|^{}_S\xi\} \cup \mbox{}\\
 && \{\sigma \mid \exists \langle\nu,X\rangle\in\faildd(p), \xi \in \infdd(q).\;
  \sigma\in \nu\|^{}_S\xi\} \cup \mbox{}\\
 && \{\sigma \mid \exists \nu\in\infdd(p),
  \langle\xi,X\rangle\in\faildd(q).\; \sigma\in \nu\|^{}_S\xi\} \cup \mbox{}\\
 && \{\sigma \in\Act^\infty \mid \forall \rho<\sigma \exists \nu\in
  \diverg(p\|^{}_S q).\; \rho\leq\nu<\sigma\}\\
\faildd(p\|^{}_S q) &=& \{\langle\sigma,X\cup Y\rangle \mid \exists
  \langle\nu,X\rangle\in\faildd(p), \langle\xi,Y\rangle\in\faildd(q).\\
 && \hfill X\setminus S = Y\setminus S \wedge
    \sigma\in \nu\|^{}_S\xi\} \cup \mbox{}\\
 && \{\langle\sigma,X\rangle \mid \sigma\in\diverg(p\|^{}_S q)
  \wedge X\subseteq\Act\}.\\
\diverg(\tau^{}_I(p)) &=& \{\tau^{}_I(\sigma) \mid \tau^{}_I(\sigma) \in A^*
  \wedge \sigma \in\infdd(p)\cup\diverg(p)\} \\
\infdd(\tau^{}_I(p)) &=& \{\tau^{}_I(\sigma) \mid \tau^{}_I(\sigma)\in A^\infty
  \wedge \sigma \in\infdd(p)\} \cup\mbox{}\\
 && \{\sigma \in\Act^\infty \mid \forall \rho<\sigma \exists \nu\in
  \diverg(\tau^{}_I(p)).\; \rho\leq\nu<\sigma\}\\
\faildd(\tau^{}_I(p)) &=& \{\langle\tau^{}_I(\sigma),X\rangle \mid
  \langle\sigma,X\cup I\rangle \in \faildd(p)\} \cup \mbox{}\\
 && \{\langle\sigma,X\rangle \mid \sigma\in\diverg(\tau^{}_I(p))
  \wedge X\subseteq\Act\}\\
\diverg(\lambda^m_s(p)) &=& \{\lambda^m_s(\sigma) \mid \sigma\in\diverg(p)\} \\
\infdd(\lambda^m_s(p)) &=& \{\lambda^m_s(\sigma) \mid
\sigma\in\infdd(p)\} \cup \mbox{}\\
 && \{\sigma \in\Act^\infty \mid \forall \rho<\sigma \exists \nu\in
  \diverg(\lambda^m_s(p)).\; \rho\leq\nu<\sigma\}\\
\faildd(\lambda^m_s(p)) &=& \{\langle\lambda^m_s(\sigma),X\rangle \mid
                \langle\sigma,\lambda^{-m}_s(X)\rangle \in\faildd(p)\} \cup \mbox{}\\
 && \{\langle\sigma,X\rangle \mid \sigma\in\diverg(\lambda^m_s(p))
		\wedge X\subseteq\Act\}.\\
\end{array}$$

``$\Rightarrow$'': Let $\sqsubseteq$ be any precongruence for $\cal O$
that respects conditional liveness properties, and suppose
$p\sqsubseteq q$. I have to establish that $p\sqsubseteq_{\it FDI}^d q$.
W.l.o.g.\ I may assume that neither $p$ nor $q$ has any trace
containing the actions $c$ or $g$. The argument for this is as in the
proof of \thm{liveness characterisation}.

Suppose $\diverg(p) \not\supseteq \diverg(q)$; say $\sigma\in
\diverg(q)\setminus \diverg(p)$.  Let $r$ be a deterministic process
such that $\ct(r)=\{\sigma c g \}$. Then each complete trace of
$p\|^{c,g} r$ that contains $c$ also contains $g$.  Here I write
$\|^{c,g}$ for \plat{$\|^{}_{A\setminus\{c,g\}}$}, the interleaving
operator that synchronises on all visible actions except $c$ and $g$.
As $\sqsubseteq$ is a precongruence, $p\sqsubseteq q$ implies
$p\|^{c,g} r \sqsubseteq q\|^{c,g} r$, and since $\sqsubseteq$
respects the canonical conditional liveness property, I obtain that
each complete trace of $q\|^{c,g} r$ that contains $c$ must also
contain $g$. However, as $\sigma\in\diverg(q)$, $\sigma c \in
\diverg(q\|^{c,g} r) \subseteq \ct(q\|^{c,g} r)$, although $\sigma c$
does not contain~$g$.

Suppose $\infdd(p) \not\supseteq \infdd(q)$; say $\sigma\in
\infdd(q)\setminus \infdd(p)$.  So $\sigma \not\in \infinite(p)$ and
there is a $\rho < \sigma$ such that $\rho\leq\rho\nu<\sigma$ for no
sequence $\rho\nu\in\diverg(p)$.  Let $r$ be a deterministic process
such that $\ct(r)=\{\rho c \nu g \mid \rho\nu < \sigma\} \cup
\{\sigma\}$.  Then each complete trace of $p\|^g r$ that contains $c$,
must also contain $g$.  As $\sqsubseteq$ is a precongruence,
$p\sqsubseteq q$ implies $p\|^{c,g} r \sqsubseteq q\|^{g,c} r$, and
since $\sqsubseteq$ respects the canonical conditional liveness
property, I obtain that each complete trace of $q\|^{c,g} r$ that
contains $c$ must also contain~$g$.  However, either
$\sigma\in\infinite(q)$ or $\rho\nu\in\diverg(q)$ for some
$\rho\leq\rho\nu<\sigma$. In each case $q\|^{c,g} r$ has a complete
trace that contains $c$ but not $g$.

Suppose $\faildd(p) \not\supseteq \faildd(q)$; say
$\langle\sigma,X\rangle\in \faildd(q)\setminus \faildd(p)$.  So
$\langle\sigma,X\rangle\not\in\failures(p)$ and
$\sigma\not\in\diverg(p)$.  Let $r$ be a deterministic process with
$\ct(r)= \{\sigma c a \mid a \in X\}$, let $C$ be the set of sequences
containing $c$, and consider the conditional liveness property given by
$C$ and $G:= \{\sigma c a \mid a \in X\}$.  Then $p\|^c r \models
\textit{liveness}_C(G)$.  As $\sqsubseteq$ is a precongruence,
$p\sqsubseteq q$ implies $p\|^c r \sqsubseteq q\|^c r$, and since
$\sqsubseteq$ respects conditional liveness properties, also $q\|^g r
\models \textit{liveness}_C(G)$.  However, either
$\langle\sigma,X\rangle\in\failures(q)$ or $\sigma\in\diverg(q)$. So
$\sigma c \in \ct(q\|^g r)$, contradicting that $q\|^c r \models
\textit{liveness}_C(G)$.
\end{proof}
In \cite{Ros04}, Bill Roscoe has shown that $\sqsubseteq_{\it FDI}^d$ is a
precongruence for all operators of CSP; he also developed a new fixed
point theory that shows that it is a congruence for recursion as well.

\section{Linear Time Properties}\label{sec-LT}

Safety, liveness, and conditional liveness properties, as studied in
the previous sections, are special cases of \emph{linear time
properties}.  A linear time property can be thought of as any
requirement on the observable content of the runs of a process. The
property is satisfied by a process when the observable content of all
its maximal runs satisfy this requirement. Hence a linear time property
can be formalised by the set of sequences over $\Act^\omega$ that,
when performed in a maximal run of a process, meet the requirement.
\begin{definition}{LT}
A \emph{linear time property} of processes in an LTS is given by a set
$P\subseteq \Act^\omega$.
A process $p$ \emph{satisfies} this property, notation
$p\models P$, when $\ct(p)\subseteq P$.
\end{definition}
A safety property is a special kind of linear time property, namely
$\emph{safety}(B) =\{\sigma\mathbin\in\Act^\omega \mid \neg \exists \rho\mathbin\in
B.\; \rho\leq\sigma\}$. Likewise, $\emph{liveness}(G)=
\{\sigma\in\Act^\omega \mid \exists \rho\in G.\; \rho\leq\sigma\}$,
and\\ $\emph{liveness}_C(G)=
\{\sigma\in\Act^\omega \mid (\exists \rho\in C.\; \rho\leq\sigma)
\Rightarrow (\exists \nu\in G.\; \nu\leq\sigma)\}$.

In \cite{AS85} and most subsequent work, liveness properties are
formalised in a different way than in this paper. For the canonical
liveness property it is fundamentally impossible to ever tell that
it is not going to be satisfied when one has only observed a finite
prefix of a maximal run of a process. For if ``something good'' is
promised to happen, it is always possible to assume it will be further
in the future. In \cite{AS85}, this is taken to be the defining
characteristic of liveness properties, and a property $P$ is called a
liveness property iff $\forall\rho\in\Act^*.\exists \sigma\in P.\;
\rho\leq \sigma$.

The property $\emph{liveness}(G)$ with $G=\{a\}$ for instance says
that the first visible action of a process should be an $a$.
It is a liveness property in my sense, since the first action being an
$a$ can be thought of as a good thing that happened eventually; here
the requirement that it has to happen as first action could be part of
one's concept of \emph{good}. However, it is not a liveness property
as formalised in \cite{AS85} and subsequent work, since the
occurrence of a $b\neq a$ as first action proves that the property
will never be satisfied.

The property that from some point onwards all visible actions a
process performs should be $g$'s, is an example of a liveness property
in the sense of \cite{AS85} that is not a liveness property in my
sense. Namely, at no point can one ever tell that something good has
happened. 

A well know theorem \cite{AS85} says that any linear time property $P$
can be written as the conjunction ${\it safety}(B) \cap P_{\it liveness}$
of a safety property and a liveness property in the sense of \cite{AS85}.
Namely,
\begin{center}
$B := \{\rho\in \Act^* \mid \neg\exists \sigma\in P.\;\rho\leq\sigma\}$
\quad and \quad $P_{\it liveness}:=P \cup (\Act^\omega\setminus\emph{safety}(B))$.
\end{center}
Such a theorem does not hold for my liveness properties.

My characterisation of $\sqsubseteq_{\it liveness}$ would still be
valid if I would have taken as class of liveness properties the
intersection of mine and the ones from \cite{AS85}. This follows
immediately from \thm{canonical liveness}, as the canonical liveness
property is in this intersection. So the extra generality in my
definition is harmless. However, the extra restriction makes a
difference, as the canonical conditional liveness property, for
instance, is a liveness property in the sense of \cite{AS85}.

Liveness properties in the sense of \cite{AS85} are studied because
proving them requires a different tool set than proving safety
properties. However, as far as practical applications are concerned,
one is mostly interested in conjunctions of safety and liveness
properties, i.e.\ general linear time properties. I will therefore not
try to characterise coarsest congruences that respect just the
liveness properties in the sense of \cite{AS85}.

The coarsest congruence respecting all linear time properties has been
characterised as \emph{NDFD-equivalence} by Roope Kaivola and Antti
Valmari in \cite{KV92}; this results extends to preorders in a
straightforward way. The NDFD preorder can be defined just like
$\sqsubseteq_{\it FDI}^d$, except that $\inf(\_\!\_)$ is used instead
of $\infd(\_\!\_)$. In fact, this result can also be obtained as
corollary of what we have seen so far.

\begin{theorem}{lt characterisation}
$p \sqsubseteq_{\textit{\scriptsize lt-properties}} q ~~\Leftrightarrow~~
\begin{array}[t]{@{}r@{~\supseteq~}l@{}}
\diverg(p) & \diverg(q) \wedge \mbox{}\\
\infinite(p) & \infinite(q) \wedge \mbox{}\\
\faildd(p) & \faildd(q).
\end{array}$
\end{theorem}

\begin{proof}
Let $\sqsubseteq_{\it NDFD}$ be the preorder defined by:
$p\sqsubseteq_{\it NDFD} q$ iff the right-hand side of \thm{lt
  characterisation} holds.

\noindent
``$\Leftarrow$'':
It suffices to establish that
$\sqsubseteq_{\it NDFD}$ is a precongruence for $\cal O$ that respects
all linear time properties.

To show that $\sqsubseteq_{\it NDFD}$ respects linear time properties,
let $P\subseteq \Act^\omega$, $p\sqsubseteq_{\it NDFD} q$, and suppose
$p \models P$. I need to show that $q \models P$. So suppose $\sigma
\in \ct(q)$.  Then either $\sigma\in\diverg(g) \subseteq\diverg(p)$ or
$\sigma \in \infinite(q) \subseteq \infinite(q)$ or
$\langle\sigma,\Act\rangle \in \failures(q) \subseteq \faildd(q)
\subseteq \faildd(p)$.  In the last case, one has either
$\langle\sigma,\Act\rangle \in \failures(p)$ or
$\sigma\in\diverg(p)$. So in all cases $\sigma\in\ct(p)$.
Since $p \models P$, it must be that $\sigma\in P$.
It follows that $\ct(q)\subseteq P$, i.e.\ $q\models P$.

That $\sqsubseteq_{\it NDFD}$ is a precongruence for $\|^{}_S$,
$\tau^{}_I$ and $\lambda^m_s$ follows from similar, but simpler,
observations as in the proof of \thm{conditional liveness characterisation}.

``$\Rightarrow$'': Let $\sqsubseteq$ be any precongruence for $\cal O$
that respects linear time properties, and suppose $p\sqsubseteq q$. I
have to establish that $p\sqsubseteq_{\it NDFD} q$.  That $\diverg(p)
\supseteq \diverg(q)$ and $\faildd(p) \supseteq \faildd(q)$ follows
immediately from \thm{conditional liveness characterisation}, using
that conditional liveness properties are linear time properties.
That $\infinite(p) \supseteq \infinite(q)$ follows immediately by
considering the linear time property $\ct(p)$.
\end{proof}
To obtain this result it suffices to define
$\sqsubseteq_{\textit{\scriptsize lt-properties}}$ as the coarsest
precongruence w.r.t.\ $\|_S$ and injective renaming that respects all
linear time properties.  However, it happens to also be a
precongruence for all operators of CSP\@.

Linear time properties do not capture the entire observable behaviour
or processes in the neutral environment. Orthogonal to them are
\emph{possibility properties}, such as: a process \emph{may} do an
action $g$. As argued by Leslie Lamport, ``verifying possibility
properties tells you nothing interesting about a system''
\cite{Lam98}. Nevertheless, it is not hard to characterise the
coarsest precongruence for $\cal O$ that respects linear time
properties as well as all possibility properties, and thereby arguably
the entire observable behaviour of a processes in a neutral
environment. It is $\equiv_{\it NDFD}$, the symmetric closure of
$\sqsubseteq_{\it NDFD}$.

\section{Concluding remark}

The methodology of the paper is close in spirit to the work on testing
equivalences by Rocco De Nicola and Matthew Hennessy \cite{DH84}, and
the results in Sections~\ref{sec-safety} and~\ref{sec-liveness} are
comparable as well. The notion of \emph{must testing} of \cite{DH84}
could be reinterpreted as a way to test liveness properties, and
hence, unsurprisingly, my preorder $\sqsubseteq_{\it liveness}$ is
exactly the must-testing preorder of \cite{DH84}. However, my safety
preorder is exactly the \emph{inverse} of the \emph{may testing}
preorder of \cite{DH84}.  This can be explained by thinking, in the
context of may testing, of the ``success''-action $\omega$ as marking
a state of \emph{failure}, rather than one of \emph{success}.  Now the
property of whether a process may reach $\omega$ is exactly the
negation of whether it will always avoid $\omega$.  This turns
may-testing around, from testing certain possibility properties, to
testing safety properties.
It remains to elaborate a theory of testing that captures the concept
of conditional liveness.

\bibliographystyle{eptcs}

\begin{thebibliography}{10}
\providecommand{\bibitemstart}[1]{\bibitem{#1}}
\providecommand{\bibitemend}{}
\providecommand{\bibliographystart}{}
\providecommand{\bibliographyend}{}
\providecommand{\url}[1]{\texttt{#1}}
\providecommand{\urlprefix}{Available at }
\providecommand{\bibinfo}[2]{#2}
\bibliographystart

\bibitemstart{AG08}
\bibinfo{editor}{M.~Alexander} \& \bibinfo{editor}{W.~Gardner}, editors
  (\bibinfo{year}{2008}): \emph{\bibinfo{title}{Process Algebra for Parallel
  and Distributed Processing}}.
\newblock \bibinfo{publisher}{Chapman \& Hall}.
\bibitemend

\bibitemstart{AS85}
\bibinfo{author}{B.~Alpern} \& \bibinfo{author}{F.B. Schneider}
  (\bibinfo{year}{1985}): \emph{\bibinfo{title}{Defining liveness}}.
\newblock {\sl \bibinfo{journal}{Information Processing Letters}}
  \bibinfo{volume}{21}, pp. \bibinfo{pages}{181--185}.
\newblock \urlprefix\url{http://dx.doi.org/10.1016/0020-0190(85)90056-0}.
\bibitemend

\bibitemstart{BB88}
\bibinfo{author}{J.C.M. Baeten} \& \bibinfo{author}{J.A. Bergstra}
  (\bibinfo{year}{1988}): \emph{\bibinfo{title}{Global renaming operators in
  concrete process algebra}}.
\newblock {\sl \bibinfo{journal}{Information and Computation}}
  \bibinfo{volume}{78}(\bibinfo{number}{3}), pp. \bibinfo{pages}{205--245}.
\bibitemend

\bibitemstart{BK85}
\bibinfo{author}{J.A. Bergstra} \& \bibinfo{author}{J.W. Klop}
  (\bibinfo{year}{1985}): \emph{\bibinfo{title}{Algebra of communicating
  processes with abstraction}}.
\newblock {\sl \bibinfo{journal}{Theoretical Computer Science}}
  \bibinfo{volume}{37}(\bibinfo{number}{1}), pp. \bibinfo{pages}{77--121}.
\bibitemend

\bibitemstart{BHR84}
\bibinfo{author}{S.D. Brookes}, \bibinfo{author}{C.A.R. Hoare} \&
  \bibinfo{author}{A.W. Roscoe} (\bibinfo{year}{1984}): \emph{\bibinfo{title}{A
  theory of communicating sequential processes}}.
\newblock {\sl \bibinfo{journal}{Journal of the ACM}}
  \bibinfo{volume}{31}(\bibinfo{number}{3}), pp. \bibinfo{pages}{560--599}.
\bibitemend

\bibitemstart{DH84}
\bibinfo{author}{R.~De~Nicola} \& \bibinfo{author}{M.~Hennessy}
  (\bibinfo{year}{1984}): \emph{\bibinfo{title}{Testing equivalences for
  processes}}.
\newblock {\sl \bibinfo{journal}{Theoretical Computer Science}}
  \bibinfo{volume}{34}, pp. \bibinfo{pages}{83--133}.
\bibitemend

\bibitemstart{GV06}
\bibinfo{author}{R.J. van Glabbeek} \& \bibinfo{author}{M.~Voorhoeve}
  (\bibinfo{year}{2006}): \emph{\bibinfo{title}{Liveness, Fairness and
  Impossible Futures}}.
\newblock In: \bibinfo{editor}{Christel Baier} \& \bibinfo{editor}{Holger
  Hermanns}, editors: {\sl \bibinfo{booktitle}{CONCUR}}, {\sl
  \bibinfo{series}{\rm LNCS}} \bibinfo{volume}{4137},
  \bibinfo{publisher}{Springer}, pp. \bibinfo{pages}{126--141}.
\newblock \urlprefix\url{http://dx.doi.org/10.1007/11817949_9}.
\bibitemend

\bibitemstart{vG01}
\bibinfo{author}{R.J.~van Glabbeek} (\bibinfo{year}{2001}):
  \emph{\bibinfo{title}{The Linear Time -- Branching Time Spectrum {I}; The
  Semantics of Concrete, Sequential Processes}}.
\newblock In: \bibinfo{editor}{J.A. Bergstra}, \bibinfo{editor}{A.~Ponse} \&
  \bibinfo{editor}{S.A. Smolka}, editors: {\sl \bibinfo{booktitle}{Handbook of
  Process Algebra}}, chapter~\bibinfo{chapter}{1},
  \bibinfo{publisher}{Elsevier}, pp. \bibinfo{pages}{3--99}.
\newblock \urlprefix\url{http://Boole.stanford.edu/pub/spectrum1.ps.gz}.
\bibitemend

\bibitemstart{KV92}
\bibinfo{author}{R.~Kaivola} \& \bibinfo{author}{A.~Valmari}
  (\bibinfo{year}{1992}): \emph{\bibinfo{title}{The Weakest Compositional
  Semantic Equivalence Preserving Nexttime-less Linear Temporal Logic}}.
\newblock In: \bibinfo{editor}{Rance Cleaveland}, editor: {\sl
  \bibinfo{booktitle}{CONCUR'92}}, {\sl \bibinfo{series}{\rm LNCS}}
  \bibinfo{volume}{630}, \bibinfo{publisher}{Springer}, pp.
  \bibinfo{pages}{207--221}.
\newblock \urlprefix\url{http://dx.doi.org/10.1007/BFb0084793}.
\bibitemend

\bibitemstart{Lam77}
\bibinfo{author}{L.~Lamport} (\bibinfo{year}{1977}):
  \emph{\bibinfo{title}{Proving the correctness of multiprocess programs}}.
\newblock {\sl \bibinfo{journal}{IEEE Transactions on Software Engineering}}
  \bibinfo{volume}{3}(\bibinfo{number}{2}), pp. \bibinfo{pages}{125--143}.
\bibitemend

\bibitemstart{Lam98}
\bibinfo{author}{Leslie Lamport} (\bibinfo{year}{1998}):
  \emph{\bibinfo{title}{Proving Possibility Properties}}.
\newblock {\sl \bibinfo{journal}{Theoretical Computer Science}}
  \bibinfo{volume}{206}(\bibinfo{number}{1-2}), pp. \bibinfo{pages}{341--352}.
\newblock \bibinfo{note}{See especially
  \url{http://research.microsoft.com/en-us/um/people/lamport/pubs/pubs.html\#l%
amport-possibility}}.
\bibitemend

\bibitemstart{Levy08}
\bibinfo{author}{P.B. Levy} (\bibinfo{year}{2008}):
  \emph{\bibinfo{title}{Infinite trace equivalence}}.
\newblock {\sl \bibinfo{journal}{Annals of Pure and Applied Logic}}
  \bibinfo{volume}{151}(\bibinfo{number}{2-3}), pp. \bibinfo{pages}{170--198}.
\newblock \urlprefix\url{http://dx.doi.org/10.1016/j.apal.2007.10.007}.
\bibitemend

\bibitemstart{Mi90ccs}
\bibinfo{author}{R.~Milner} (\bibinfo{year}{1990}):
  \emph{\bibinfo{title}{Operational and algebraic semantics of concurrent
  processes}}.
\newblock In: \bibinfo{editor}{J.~van Leeuwen}, editor: {\sl
  \bibinfo{booktitle}{Handbook of Theoretical Computer Science}},
  chapter~\bibinfo{chapter}{19}, \bibinfo{publisher}{Elsevier Science
  Publishers B.V. (North-Holland)}, pp. \bibinfo{pages}{1201--1242}.
\newblock \bibinfo{note}{Alternatively see{ \em Communication and Concurrency},
  Prentice-Hall, Englewood Cliffs, 1989, of which an earlier version appeared
  as{ \em A Calculus of Communicating Systems}, LNCS 92, Springer-Verlag,
  1980}.
\bibitemend

\bibitemstart{Pu01}
\bibinfo{author}{A.~Puhakka} (\bibinfo{year}{2001}):
  \emph{\bibinfo{title}{Weakest Congruence Results Concerning ``Any-Lock''}}.
\newblock In: \bibinfo{editor}{N.~Kobayashi} \& \bibinfo{editor}{B.C. Pierce},
  editors: {\sl \bibinfo{booktitle}{Proceedings TACS'01}}, {\sl
  \bibinfo{series}{\rm LNCS}} \bibinfo{volume}{2215},
  \bibinfo{publisher}{Springer}, pp. \bibinfo{pages}{400--419}.
\bibitemend

\bibitemstart{Ros04}
\bibinfo{author}{A.~W. Roscoe} (\bibinfo{year}{2004}):
  \emph{\bibinfo{title}{Seeing Beyond Divergence}}.
\newblock In: \bibinfo{editor}{Ali~E. Abdallah}, \bibinfo{editor}{Cliff~B.
  Jones} \& \bibinfo{editor}{Jeff~W. Sanders}, editors: {\sl
  \bibinfo{booktitle}{25 Years Communicating Sequential Processes}}, {\sl
  \bibinfo{series}{Lecture Notes in Computer Science}} \bibinfo{volume}{3525},
  \bibinfo{publisher}{Springer}, pp. \bibinfo{pages}{15--35}.
\newblock \urlprefix\url{http://dx.doi.org/10.1007/11423348_2}.
\bibitemend

\bibitemstart{Ros97}
\bibinfo{author}{A.W. Roscoe} (\bibinfo{year}{1997}): \emph{\bibinfo{title}{The
  Theory and Practice of Concurrency}}.
\newblock \bibinfo{publisher}{Prentice-Hall}.
\newblock
  \urlprefix\url{http://www.comlab.ox.ac.uk/bill.roscoe/publications/68b.pdf}.
\bibitemend

\bibliographyend
\end{thebibliography}

\end{document}